\documentclass[10pt,conference]{./IEEEtran}
\usepackage{cite}
\usepackage{color}
\usepackage{amsmath}
\usepackage{wasysym,mathrsfs,marvosym,soul}
\usepackage{latexsym,amssymb,amsmath,empheq}
\usepackage{graphicx}
\usepackage{psfig}
\usepackage{psfrag}
\usepackage{subfigure}
\usepackage{url}
\usepackage{array}
\usepackage{fontenc}
\makeatletter

\DeclareSymbolFont{rsfs}{U}{rsfs}{m}{n}
\DeclareSymbolFontAlphabet{\mathscr}{rsfs}

\newcommand{\mscr}{\mathscr}
\newcommand{\dv}{\mathbf} 
\newcommand{\mc}{\mathcal} 

\newtheorem{theorem}{Theorem}
\newtheorem{lemma}{Lemma}
\newtheorem{definition}{Definition}

\begin{document}
\title{Rate Regions for the Partially-Cooperative Relay Broadcast Channel with Non-causal Side Information}

\author{
\authorblockN{Abdellatif Zaidi}
\authorblockA{Communication and Remote Sensing Lab. (TELE)\\
Universit\'e Catholique de Louvain\\
Louvain-La-Neuve, 1348, Belgium\\
Email: Abdellatif.Zaidi@ensta.org}
\and
\authorblockN{Luc Vandendorpe}
\authorblockA{Communication and Remote Sensing Lab. (TELE)\\
Universit\'e Catholique de Louvain\\
Louvain-La-Neuve, 1348, Belgium\\
Email: Luc.Vandendorpe@tele.ucl.ac.be}
}
\bibliographystyle{IEEEtran}
\maketitle

\begin{abstract}
In this work, we consider a partially cooperative relay broadcast channel (PC-RBC) controlled by random parameters. We provide rate regions for two different situations: 1) when side information (SI) $S^n$ on the random parameters is non-causally known at both the source and the relay and, 2) when side information $S^n$ is non-causally known at the source only. These achievable regions are derived for the general discrete memoryless case first and then extended to the case when the channel is degraded Gaussian and the SI is additive i.i.d. Gaussian. In this case, the source uses generalized dirty paper coding (GDPC), i.e., DPC combined with partial state cancellation, when only the source is informed, and DPC alone when both the source and the relay are informed. It appears that, even though it can not completely eliminate the effect of the SI (in contrast to the case of source and relay being informed), GDPC is particularly useful when only the source is informed.
\end{abstract}
\section{Introduction}\label{secI}
A three-node relay broadcast channel (RBC) is a communication network where a source node transmits both common information and private information sets to two destination nodes, destination $1$ and destination $2$, that cooperate by exchanging information. This may model "downlink" communication systems that exploit relaying and user cooperation to improve reliability and throughput. In this work, we consider the RBC in which only one of the two destinations (e.g., destination $1$) assists the other destination. This channel is referred to as {\it partially cooperative RBC} (PC-RBC) \cite{LV05,LK06}. Moreover, we assume that the channel is controlled by random parameters and that side information $S^n$ on these random parameters is non-causally known either at both the source and destination $1$ (i.e., the relay) (we refer to this situation as {\it PC-RBC with informed source and relay}) or at the source only (we refer to this situation as {\it PC-RBC with informed source only}). The random state may represent random fading, interference imposed by other users, etc. (see \cite{BPS98} for a comprehensive overview on state-dependent channels). The PC-RBC under investigation is shown in Fig.~\ref{RelayBroadcastChannelWithState}. It includes the standard relay channel (RC) as a special case, when no private information is sent to destination $1$, which then simply acts as relay for destination $2$.\\ 
\begin{figure}[htpb]
\centering
\psfrag{source}[][][0.75]{Tx}
\psfrag{destination}[][][0.75]{Rx $2$}
\psfrag{node2}[][][0.75]{ Relay}
\psfrag{a}[][][0.75]{\bf A}
\psfrag{b}[][][0.75]{\bf B}
\psfrag{Rx2}[][][0.75]{Rx $1$}
\psfrag{input}[][][0.5]{$(W_0,W_1,W_2) \:\:\qquad$}
\psfrag{input_node2}[][][0.75]{$Y_1$}
\psfrag{output_node2}[][][0.75]{$\:\:X_2$}
\psfrag{output_source}[][][0.75]{$\:\:X_1$}
\psfrag{received_signal}[][][0.75]{$Y_2$}
\psfrag{state}[][][0.75]{$S^n$}
\psfrag{and}[][][0.75]{$\&$}
\psfrag{channel_distribution}[][][0.75]{$p(y_1,y_2|x_1,x_2,s)$}
\psfrag{output}[][][0.75]{$\qquad (\hat{W}_0,\hat{W}_2)$}
\includegraphics[width=0.9\linewidth]{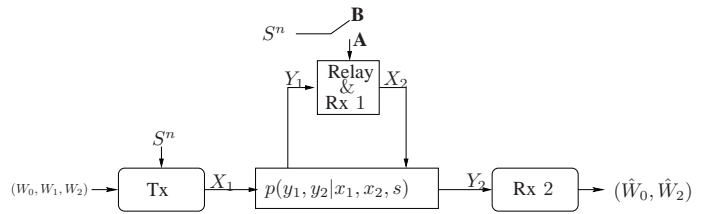}
\caption{Partially-cooperative relay broadcast channel (PC-RBC) with state information $S^n$ non-causally known either at both the source and the relay (A) or at the source only (B).}
\label{RelayBroadcastChannelWithState}
\end{figure}
For the discrete memoryless PC-RBC with informed source and relay (Section \ref{secII}), we derive an achievable rate region based on the relay operating in the decode-forward (DF) scheme. We also show that this region is tight and provides the full capacity region when the channel outputs are corrupted by degraded Gaussian noise terms and the SI $S^n$ is additive i.i.d. Gaussian (referred to as {\it D-AWGN partially cooperative RBC}). Similarly to \cite{C83,KSS04}, it appears that, in this case, the SI does not affect the capacity region, even though destination $2$ has no knowledge of the state. The result on the property that a known additive state does not affect capacity (as long as full knowledge of this state is available at the transmitter) has been initially established for single-user Gaussian channel in \cite{C83}, and then extended to some other multi-user Gaussian channels in \cite{KSS04}.

For the PC-RBC with informed source only (Section \ref{secIII}), we derive achievable rate regions for the discrete memoryless and the D-AWGN memoryless cases, based on the relay operating in DF. The D-AWGN  case uses generalized dirty paper coding (GDPC), which allows arbitrary (negative) correlation between codewords and the SI, at the source. In this case, we show that, even though the relay is uninformed, it benefits from the availability of the SI at the source, which then helps the relay by allocating a fraction of its power to cancel the state, and uses the remaining of its power to transmit pure information using DPC. However, even though this region is larger than that obtained by DPC alone (i.e., without partial state cancellation), the effect of the state can not be completely canceled as in the case when both the source and the relay are informed.

The results in this paper readily apply to the standard relay channel (RC), as a special case of a PC-RBC when no private information is sent to destination $1$. More generally, they shed light on cooperation between informed and uninformed nodes and can in principle be extended to channels with many cooperating nodes, with only a subset of them being informed. Section \ref{secIV} gives an illustrative numerical example. Section \ref{secV} draw some concluding remarks. Proofs are relegated to Section \ref{secV}.

\section{Partially-Cooperative RBC with Informed Source and Relay}\label{secII}
Consider the channel model for the discrete memoryless PC-RBC with informed source and relay denoted by $\{\mscr X_1{\times}\mscr X_2,p(y_1,y_2|x_1,x_2,s),\mscr Y_1\times\mscr Y_2,\mscr S\}$ and depicted in Fig.\ref{RelayBroadcastChannelWithState}. It consists of a source with input $X_1$, a relay with input $X_2$, a state-dependent probability distribution $p(y_1,y_2|x_1,x_2,s)$ and two channel outputs $Y_1$ and $Y_2$ at destinations $1$ (the relay) and $2$, respectively. The source sends a common message $W_0$ that is decoded by both destinations and private messages $W_1$ and $W_2$ that are decoded by destinations $1$ and $2$, respectively.

In this section, we consider the scenario in which the PC-RBC is embedded in some environment with SI $S^n$ available non-causally at both the source and the relay. We assume that $S_i$'s are i.i.d. random variables $\sim p(s)$, $i=1,\hdots,n$, and that the channel is memoryless.
\subsection{Inner bound on capacity region}\label{secII_subsecB}
The following Lemma gives an inner bound on capacity region for the PC-RBC with informed source and relay, based on the relay operating in the decode-and-forward (DF) scheme.

\begin{lemma}\label{AchievableRatePartiallyCooperativeRBCWithState}
For a discrete memoryless partially cooperative relay broadcast channel $p(y_1,y_2|x_1,x_2,s)$ with state information $S^n$ non-causally available at the source and destination $1$ (which also acts as a relay for destination $2$) but not at destination $2$, a rate tuple $(R_0,R_1,R_2)$ is achievable if
\begin{align}
&R_1 \: < \:\: I(X_1;Y_1|SU_1X_2),\nonumber\\
&R_0+R_2 < \min \: \Big\{I(U_2;Y_1|SU_1),I(U_1U_2;Y_2)-I(U_1U_2;S)\Big\},
\label{AchievableRegionTheorem1}
\end{align}
for some joint distribution of the form $$p(s)p(u_1,u_2,x_1,x_2|s)p(y_1|x_1,x_2,s)p(y_2|y_1,x_2),$$ where $U_1$ and $U_2$ are auxiliary random variables with finite cardinality bounds.
\end{lemma}
The proof is similar to that, given in Section \ref{secVI}, for Lemma \ref{AchievableRatePartiallyCooperativeRBCInformedSource} (see below). However, it is more lengthy. We omitted it here for brevity. 
\subsection{D-AWGN Partially Cooperative RBC}\label{secII_subsecC}
We now assume that the state is additive i.i.d. Gaussian. Furthermore, we assume that the channel outputs are corrupted by degraded Gaussian noise terms. We refer to this channel as the D-AWGN PC-RBC with informed source and relay,  meaning that there exist random variable $Z_1 \sim \mc N(0,N_1)$ and $Z'_2 \sim \mc N(0,N_2-N_1)$ with $N_1 < N_2$, independent of each other and independent of the state $S^n$, such that
\begin{align}
Y_1&=X_1+S+Z_1,\nonumber\\
Y_2&=Y_1+X_2+Z'_2.
\label{DegradedGaussianPC-RBCWithState}
\end{align}
 The channel input sequences $\{x_{1,n}\}$ and $\{x_{2,n}\}$ are subject to power constraints $P_1$ and $P_2$, respectively, i.e., $\sum_{i=1}^{n}x_{1i}^2 \leq nP_1$ and  $\sum_{i=1}^{n}x_{2i}^2 \leq nP_2$; and the state $S^n$ is distributed according to $\mc N(0,QI)$.

The D-AWGN PC-RBC with no state has been introduced and studied in \cite{LV05}. It has been shown that its capacity region is given by the region with the rate tuples $(R_0,R_1,R_2)$ satisfying \cite{LV05}
\begin{subequations}
\begin{align}
\label{Rate1R1D-AWGNPartiallyCooperativeRBCWithoutState}
&R_1 \: \:<\:\: C\Big(\frac{{\gamma}P_1}{N_1}\Big)\\
\label{Rate2R2D-AWGNPartiallyCooperativeRBCWithoutState}
&R_0+R_2 < \max_{\beta}\: \min  \bigg\{C(\frac{{\beta}\bar{\gamma}P_1}{{\gamma}P_1+N_1}),\nonumber\\
&\hspace{3cm}C\Big(\frac{\bar{\gamma}P_1+P_2+2\sqrt{\bar{\beta}\bar{\gamma}P_1P_2}}{{\gamma}P_1+N_2}\Big)\bigg\},
\end{align}
\label{CapacityRegionD-AWGNPartiallyCooperativeRBCWithoutState}
\end{subequations}
for some $\gamma \in [0,1]$, where $\bar{\gamma}=1-\gamma$, $\bar{\beta}=1-\beta$ and $C(x):=0.5\log_2(1+x)$.

We now turn to the case when there is an additive i.i.d. SI $S^n$ which is non-causally known to both the source and destination $1$ (the relay) but not to destination $2$. We obtain the following result, similar in nature (and in proof) to that provided for a physically degraded Gaussian RC in \cite[Theorem 3]{KSS04}.

\begin{theorem}\label{PPED-AWGNPartiallyCooperativeRBC}
The capacity region of the D-AWGN Partially Cooperative Relay Broadcast Channel with state information non-causally available at the source, destination $1$ (the relay) but not destination $2$ is given by the standard capacity (\ref{CapacityRegionD-AWGNPartiallyCooperativeRBCWithoutState}).
\end{theorem}

\begin{proof}
Similarly to Costa's approach \cite{C83}, we need only prove the achievability of the region, which follows  by evaluating the region (\ref{AchievableRegionTheorem1}) with the input distribution given by (\ref{AuxiliaryRandomVariablesForRBWithState}). Note that region (\ref{AchievableRegionTheorem1}) has been established for the discrete memoryless case but it can be extended to memoryless channels with discrete time and continuous alphabets using standard techniques \cite{G68}. The choice of $p(u_1,u_2,x_1,x_2|s)$ is given by 
\begin{subequations}
\begin{align}
\label{AuxiliaryRandomVariablesU1U2}
&U_1 \sim \mc N(\alpha_1S,P^{(1)}),U_2 \sim \mc N({\alpha}_2S,P^{(2)})\\
&X_2=(1-\lambda)(U_1-\alpha_1S),\:\: \lambda=\frac{\sqrt{\bar{\beta}\bar{\alpha}P_1}}{\sqrt{P^{(1)}}},\\
&X_1' \sim \mc N(0,{\gamma}P_1),\\
&X_1=\lambda(U_1-\alpha_1S)+(U_2-\alpha_2S)+X_1',
\end{align}
\label{AuxiliaryRandomVariablesForRBWithState}
\end{subequations}
where $P^{(1)}=(\sqrt{\bar{\beta}\bar{\alpha}P_1}+\sqrt{P_2})^2$, $P^{(2)}={\beta}\bar{\alpha}P_1$ and
\begin{align}
&\alpha_k=\frac{P^{(k)}}{P^{(1)}+P^{(2)}+({\alpha}P_1+N_2)},\:\:k=1,2.\nonumber
\end{align}
Furthermore, we let $X'_1$ be independent of $U_1$, $U_2$ and the state $S$.
\end{proof}
A (more intuitive) alternative approach is as follows. The source uses superposition coding to send the information intended for destination $1$, on top of that intended for destination $2$ (and carried through the relay). We decompose the source input $X_1$ into two parts, $X_1'$ with power ${\alpha}P_1$ (stands for the information intended for destination $1$), and $U$ with power $\bar{\alpha}P_1$ (stands for the information intended for destination $2$), i.e., $X_1=X_1'+U$. For the transmission of $U$, both the source and destination $1$ know the state $S^n$ and cooperate over a relay channel (considering $X'_1$ as noise) to achieve the rate (\ref{Rate2R2D-AWGNPartiallyCooperativeRBCWithoutState}) \cite{KSS04}. Next, to decode its own message, destination $1$ first peals $S$ and $U$ to make the channel $Y_1$ equivalent to $Y_1'=X_1'+Z_1$. This gives us the rate (\ref{Rate1R1D-AWGNPartiallyCooperativeRBCWithoutState}) for message $W_1$.

\section{Partially-Cooperative RBC with Informed Source Only}\label{secIII}
In this section, we assume that only the source non-causally knows the SI $S^n$.
\subsection{Discrete memoryless PC-RBC}\label{secIII_subsecA}
The following Lemma gives an inner bound on capacity region for the PC-RBC with informed source only. The result is based on the relay operating in the DF scheme.

\begin{lemma}\label{AchievableRatePartiallyCooperativeRBCInformedSource}
For a discrete memoryless partially cooperative relay broadcast channel $p(y_1,y_2|x_1,x_2,s)$ with state information $S^n$ non-causally available at the source only, a rate tuple $(R_0,R_1,R_2)$ is achievable if
\begin{align}
&R_1 \: < \:\: I(U_1;Y_1|U_2X_2)-I(U_1;S|U_2X_2)\nonumber\\
&R_0+R_2 < \min \: \Big\{I(U_2;Y_1|X_2)-I(U_2;S|X_2),\nonumber\\
&\hspace{3.5cm} I(U_2X_2;Y_2)-I(U_2;S|X_2)\Big\},
\label{AchievableRegionPCRBCWithInformedSource}
\end{align}
for some joint distribution of the form $$p(s)p(u_1,u_2,x_1,x_2|s)p(y_1|x_1,x_2,s)p(y_2|y_1,x_2),$$ where $U_1$ and $U_2$ are auxiliary random variables with finite cardinality bounds.\\
\end{lemma}
The proof is based on a combination of sliding-window \cite{K78,C82}, superposition-coding \cite{C72} and Gelfand and Pinsker's binning \cite{GP80}. See Section \ref{secV} for an outline of it.
\subsection{D-AWGN Partially Cooperative RBC}\label{secIII_subsecB}
Assume now that the PC-RBC with informed source only is degraded Gaussian,i.e.,
the channel outputs can be written as
\begin{align}
Y_1&=X_1+S+Z_1,\nonumber\\
Y_2&=Y_1+X_2+Z'_2,
\label{DegradedGaussianPC-RBCWithInformedSource}
\end{align}
where $Z_1 \sim \mc N(0,N_1)$ and $Z'_2 \sim \mc N(0,N_2-N_1)$, with $N_1 < N_2$, are independent of each other and independent of the state $S^n \sim \mc N(0,QI)$; and the input sequences $\{x_{1,n}\}$ and $\{x_{2,n}\}$ are subject to average power constraints $P_1$ and $P_2$, respectively.

We obtain an inner bound on capacity region by having the source using a generalized dirty paper coding (GDPC), which allows arbitrary (negative) correlation between the codeword and the SI and can be viewed as a partial state cancellation \cite{KL04}.

\begin{definition} Let
\begin{align}
Q'(\gamma,\rho)&:=(\sqrt{Q}-\sqrt{\rho\bar{\gamma}P_1})^2,\nonumber\\
A(\gamma,\rho,\beta,\alpha)&:=(1-\beta^2)\bar{\rho}\bar{\gamma}P_1\Big((1-\beta^2)\bar{\rho}\bar{\gamma}P_1\nonumber\\
&\hspace{3cm}+Q'(\gamma,\rho)+{\gamma}P_1+N_1\Big),\nonumber\\
B(\gamma,\rho,\beta,\alpha)&:=(1-\alpha)^2(1-\beta^2)\bar{\rho}\bar{\gamma}P_1Q'(\gamma,\rho)\nonumber\\
&+(N_1+{\gamma}P_1)\Big((1-\beta^2)\bar{\rho}\bar{\gamma}P_1+\alpha^2Q'(\gamma,\rho)\Big),\nonumber\\
C(\gamma,\rho,\beta,\alpha)&:=(1-\beta^2)\bar{\rho}\bar{\gamma}P_1\Big(\bar{\rho}\bar{\gamma}P_1+P_2\nonumber\\
&\hspace{0.5cm}+Q'(\gamma,\rho)+2\beta\sqrt{\bar{\rho}\bar{\gamma}P_1P_2}+{\gamma}P_1+N_2\Big),\nonumber\\
D(\gamma,\rho,\beta,\alpha)&:=(1-\alpha)^2(1-\beta^2)\bar{\rho}\bar{\gamma}P_1Q'(\gamma,\rho)\nonumber\\
&+(N_2+{\gamma}P_1)\Big((1-\beta^2)\bar{\rho}\bar{\gamma}P_1+\alpha^2Q'(\gamma,\rho)\Big),\nonumber\\
r_1(\gamma,\rho,\beta,\alpha)&:=\frac{1}{2}\log_2\left(\frac{A(\gamma,\rho,\beta,\alpha)}{B(\gamma,\rho,\beta,\alpha)}\right),\nonumber\\
r_2(\gamma,\rho,\beta,\alpha)&:=\frac{1}{2}\log_2\left(\frac{C(\gamma,\rho,\beta,\alpha)}{D(\gamma,\rho,\beta,\alpha)}\right),\nonumber
\end{align}
for given $0 \leq \gamma \leq 1$, $0 \leq \rho \leq \min\{1,\frac{Q}{\bar{\gamma}P_1}\}$, $0 \leq \beta \leq 1$, $0 \leq \alpha \leq 1$ and where $\bar{\gamma}=1-\gamma$ and $\bar{\rho}=1-\rho$.
\end{definition}

The following theorem gives an inner bound on capacity region for D-AWGN partially cooperative RBC with informed source only.

\begin{theorem}\label{AchievableReginDAWGNPCRBCInformedSource}
Let $\mc R^{in}(\gamma)$ be the set of all rate tuples $(R_0,R_1,R_2)$ satisfying
\begin{subequations}
\begin{align}
R_1 & \leq  \frac{1}{2}\log_2(1+\frac{{\gamma}P_1}{N_1})\\
R_0+R_2 & \leq \max_{\alpha_2,\beta,\rho} \:\:\min \:\: \Big\{r_1(\gamma,\rho,\beta,\alpha_2),r_2(\gamma,\rho,\beta,\alpha_2)\Big\},
\end{align}
\label{AchievableRate_DAWGN_PCRBC_InformedSourceOnly}
\end{subequations}
for some $0 \leq \gamma \leq 1$, where maximization is over $0 \leq \rho \leq \min\{1,\frac{Q}{\bar{\gamma}P_1}\}$, $0\leq \alpha_2 \leq 1$ and $0 \leq \beta \leq 1$. Then, $\mc R^{in}(\gamma)$ is contained in capacity region of the D-AWGN PC-RBC (\ref{DegradedGaussianPC-RBCWithInformedSource}), where state information $S^n$ is non-causally available at the source only.
\end{theorem}
\begin{proof}
The source uses a combination of superposition coding and generalized DPC. More specifically, we decompose the source input $X_1$ as
\begin{subequations}
\begin{align}
X_1&=X'_1+U,\\
\label{FormOfGDPC}
U&=-\sqrt{\frac{\rho\bar{\gamma}P_1}{Q}}S+U_{w},
\end{align}
\end{subequations}
where $X'_1$ (of power ${\gamma}P_1$), $U_w$ (of power $\bar{\rho}\bar{\gamma}P_1$) and $S$ are independent, and $\mathbb{E}[U_wX_2]=\beta\sqrt{\bar{\rho}\bar{\gamma}P_1P_2}$.
With this choice of input signals, channels $Y_1$ and $Y_2$ in (\ref{DegradedGaussianPC-RBCWithInformedSource}) become
\begin{subequations}
\begin{align}
&Y'_1=X'_1+U_w+S'+Z_1\\
&Y'_2=U_w+X_2+S'+X'_1+Z_1+Z'_2,
\end{align}
\label{EquivalentChannel}
\end{subequations}
where the Gaussian state $S'=(1-\sqrt{\frac{\rho\bar{\gamma}P_1}{Q}})S$ is known to the source and has power $Q'(\rho,\gamma)=(\sqrt{Q}-\sqrt{\rho\bar{\gamma}P_1})^2$. Then, given that the result of Lemma \ref{AchievableRatePartiallyCooperativeRBCInformedSource} which has been established for the discrete memoryless case can be extended to memoryless channels with discrete time and continuous alphabets using standard techniques \cite{G68}, the proof of achievability follows by evaluating the region (\ref{AchievableRegionPCRBCWithInformedSource}) (in which $Y_1$, $Y_2$ and $S$ are replaced by $Y'_1$, $Y'_2$ and $S'$, respectively) with the following choice of input distribution:
\begin{subequations}
\begin{align}
&U_1 \sim \mc N({\alpha}_1(1-\alpha_2)S',{\gamma}P_1),\\
&U_2 \sim \mc N({\alpha}_2S',\bar{\rho}\bar{\gamma}P_1),\\
&X_2 \sim \mc N(0,P_2),\\
&X_1=U_1+U_2-(\alpha_1+\alpha_2-\alpha_1\alpha_2+\frac{\sqrt{\rho\bar{\gamma}P_1}}{\sqrt{Q}-\sqrt{\rho\bar{\gamma}P_1}})S',
\end{align}
\label{InputDistributionForDAWGN_PCRBC_InformedSource}
\end{subequations}
where $\alpha_1={\gamma}P_1/({\gamma}P_1+N_1)$ and $0 \leq \alpha_2 \leq 1$. Furthermore, we let $\mathbb{E}[U_wX_2]=\beta\sqrt{\bar{\rho}\bar{\gamma}P_1P_2}$ and choose $X'_1$, $X_2$ and $S'$ to be independent. Through straight algebra which is omitted for brevity, it can be shown that (\ref{InputDistributionForDAWGN_PCRBC_InformedSource}) achieve the rates in (\ref{AchievableRate_DAWGN_PCRBC_InformedSourceOnly}) to complete the proof.
\end{proof}
The intuition  for (\ref{InputDistributionForDAWGN_PCRBC_InformedSource}) is as follows. Consider the channel (\ref{EquivalentChannel}). The source allocates a fraction ${\gamma}P_1$ of its power to send message $W_1$ (input $X'_1$) to destination $1$ and the remaining power, $\bar{\gamma}P_1$, to send message $W_2$ (input $U$) to destination $2$, through the relay. However, since the relay does not know the state $S^n$, the source allocates a fraction $\rho$ ($0 \leq \rho \leq \min\{1,\frac{Q}{\bar{\gamma}P_1}\}$) of the power $\bar{\gamma}P_1$ to cancel the state so that the relay can benefit from this cancellation. Then, it uses the remaining power, $\bar{\rho}\bar{\gamma}P_1$, for pure information transmission (input $U_w$).

For the transmission of message $W_2$ to destination $2$, we treat the interference $X'_1$ combined with the channel noise $Z_1+Z'_2$ as an unknown Gaussian noise. Hence, the source uses a DPC 
\begin{subequations}
\begin{align}
\label{U2_DAWGN_PCRBC_InformedSource}
&U_2 \sim \mc N({\alpha}_2S',\bar{\rho}\bar{\gamma}P_1),\\
&U_w=U_2-\alpha_2S'.
\end{align}
\end{subequations}
Furthermore, the relay can decode $U_2=U_w+\alpha_2S'$ and peal it of to make the channel to the relay equivalent to
\begin{align}
Y'_1&=Y_1-U_2=X'_1+(1-\alpha_2)S'+Z_1.
\end{align}
Thus, for the transmission of message $W_1$ to destination $1$, the source uses another DPC
\begin{subequations}
\begin{align}
\label{U1_DAWGN_PCRBC_InformedSource}
&U_1 \sim \mc N({\alpha}_1(1-\alpha_2)S',{\gamma}P_1),\\
&X'_1=U_1-\alpha_1(1-\alpha_2)S',
\end{align}
\end{subequations}
where $(1-\alpha_2)S'$ is the known state and $\alpha_1={\gamma}P_1/({\gamma}P_1+N_1)$. This gives us the rate $\frac{1}{2}\log_2(1+\frac{{\gamma}P_1}{N_1})$ for rate $R_1$.

{\it Remark 1~:} Here, we have used in essence two superimposed DPCs, with one of them being generalized. The first approach which suggests itself and which consists in using two standard (not generalized) DPCs corresponds to the special case of $\rho=0$. Also, note that, for the GDPC, there is no loss in restricting the correlation (between the source input $U$ and the state $S$) to have the form in (\ref{FormOfGDPC}), in this case.

{\it Remark 2~:} A straightforward outer bound for the capacity region of the D-AWGN partially-cooperative RBC with only the source being informed is given by (\ref{CapacityRegionD-AWGNPartiallyCooperativeRBCWithoutState}), for this is the capacity region of the D-AWGN PC-RBC without state or with state known everywhere.

{\it Remark 3~:}  The results of Lemmas \ref{AchievableRatePartiallyCooperativeRBCWithState} and \ref{AchievableRatePartiallyCooperativeRBCInformedSource} and Theorems \ref{PPED-AWGNPartiallyCooperativeRBC} and \ref{AchievableReginDAWGNPCRBCInformedSource} specialize to the relay channel (RC), by letting destination $1$ decode no private message (i.e., $R_1$=0). For the case of a RC with informed source and relay, this gives us the achievability of the rate
\begin{align}
&R=\max_{p(u_1,u_2,x_1,x_2|s)} \:\: \min \:\: \Big\{I(U_1;Y_1|SU_2),I(U_1U_2;Y_2)\nonumber\\
&\hspace{5.6cm}-I(U_1U_2;S)\Big\}.
\label{AchievableRateRCWithState}
\end{align}
Note that, even thought this rate is in general smaller than the one given in \cite[Lemma 3]{KSS04} (in which $I(U_1;Y_1|SX_2)$ is used instead of $I(U_1;Y_1|SU_2)$  in (\ref{AchievableRateRCWithState})), the two rates coincide in the Gaussian (not necessarily physically degraded) case. To see that, note that in the Gaussian case, $X_2$ is a linear combination of $U_2$ and $S$ \cite{C83}, and hence $I(U_2S;Y_1)=I(X_2S;Y_1)$. Then, writing
\begin{subequations}
{\footnotesize
\begin{align}
I(U_1U_2SX_2;Y_1)&=I(X_2S;Y_1)+I(U_1;Y_1|SX_2)+I(U_2;Y_1|SX_2U_1),\nonumber\\
                 &=I(U_2S;Y_1)+I(U_1;Y_1|SU_2)+I(X_2;Y_1|SU_1U_2),\nonumber
\end{align}}
\end{subequations}
and noticing that $I(X_2;Y_1|SU_1U_2)=0$ (since $p_{X_2|U_2S}=0,1$) and $I(U_2;Y_1|SX_2U_1)=0$ (since $(U_1,U_2) \ominus (X_1,X_2,S)\ominus (Y_1,Y_2)$ forms a Markov chain under the specified distribution in (\ref{AchievableRateRCWithState})), we get $I(U_1;Y_1|SX_2)=I(U_1;Y_1|SU_2)$.
\section{ Numerical Example}\label{secIV}
This section illustrates the achievable rate regions for D-AWGN PC-RBC and physically degraded Gaussian RC, with the help of an example. We illustrate the effect of applying GDPC in improving the throughput when only the source is informed. 

Fig.\ref{CapacityRegionPCRBCWithInformedTransmitter} depicts the inner bound using generalized DPC in Theorem \ref{AchievableRatePartiallyCooperativeRBCInformedSource}. Also shown for comparison are: an inner bound using DPC alone (i.e., GDPC with $\rho=0$) and an outer bound, obtained by assuming both the source and the relay being informed. Rate curves are depicted for both D-AWGN PC-RBC and physically degraded Gaussian RC. We see that even though the state is known only at the source, both the source and the relay benefit. 

For the physically degraded Gaussian RC, the improvement is mainly visible at high $\text{SNR}=P_1/N_1$ [dB]. This is because, the relay being operating in DF, cooperation between the source and the relay is {\it more efficient} at high $\text{SNR}$. In such range of $\text{SNR}$, capacity of the degraded Gaussian RC is driven by the amount of information that the source and the relay can, together, transfer to the destination (given by the term $I(X_1X_2;Y_2)$ in the capacity of the degraded RC). At small $\text{SNR}$ however, capacity of the degraded Gaussian RC is constrained by the broadcast bottleneck (term $I(X_1;Y_2|X_2)$). Hence, in such range of $\text{SNR}$, there is no need for the source to assist the relay by (partially) cancelling the state for it (since this would be accomplished at the cost of the power that can be allocated to transmit information from the source to the relay). An alternative interpretation is as follows. At high $\text{SNR}$, the source and the relay form two fictitious users (with only one of them being informed) sending information to same destination, over a MAC. The sum rate over this MAC is more enlarged (by the use of GDPC) at high $\text{SNR}$. This interpretation conforms with the result in \cite{KL04} for a MAC with only one informed encoder. However, note this interpretation deviates from \cite{KL04}, in that the fictitious MAC considered  here has correlated inputs).

For the D-AWGN PC-RBC, we see that both destination $1$ and destination $2$ benefits from using GDPC at the source. This can be easily understood as follows. Since applying GDPC at the source improves rate $R_2$ for destination $2$ (w.r.t. using DPC alone), the source needs lesser power, for the same amount of information to be transmitted to destination $2$ (i.e., for the same $R_2$). Hence, the power put aside can be used to increase rate $R_1$ (see the zoom on the top left of Fig.~\ref{CapacityRegionPCRBC}).
\begin{figure}
\centering
\subfigure[D-AWGN Partially Cooperative RBC]
{
\label{CapacityRegionPCRBC}
\includegraphics[width=\linewidth]{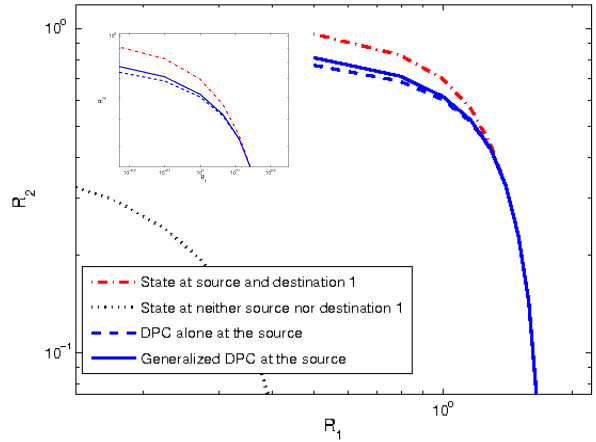}
}
\subfigure[Physically Degraded Gaussian Relay Channel]
{
\includegraphics[width=\linewidth]{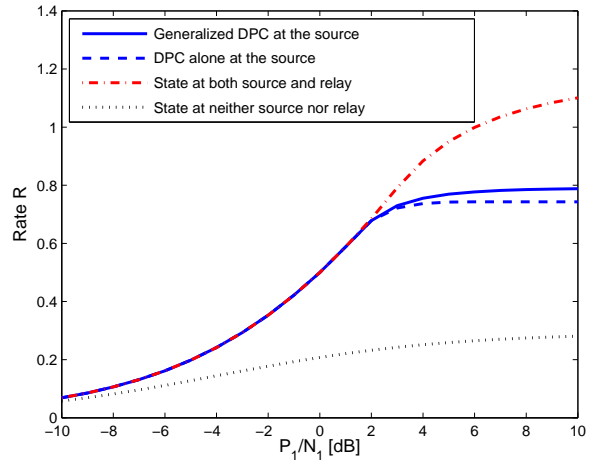}
}
\caption{Achievable rate regions for D-AWGN PC-RBC and physically degraded Gaussian RC. (a) $P_1=P_2=1=Q$, $N_1=10N_2=1$. (b) $P_1=P_2=1$, $Q=2$, $N_2=1$.}
\label{CapacityRegionPCRBCWithInformedTransmitter}
\end{figure}
\section{Concluding remarks}\label{secVI}
In many practical communication systems that exploit node cooperation to increase throughput or improve reliability, different (possibly not co-located) cooperating nodes rarely have access to the same state information (SI) about the channel (interference, fading, etc.). In this case, a more general approach to address node cooperation in such channels is to consider different SI at the different nodes. Also, as these nodes rarely have the ability to measure directly, or estimate, the channel state, a more involved approach would be to account for the cost of conveying SI (e.g., by a third party) to the different nodes (as already done for MAC, in \cite{CS05}). In this paper, we have considered the basic three-node network in which two nodes transmit information over a partially cooperative relay broadcast channel (PC-RBC). We investigated two different situations: when both the source and the relay non-causally know the channel state and, when only the source knows the state. One important finding in the latter case is that, in the degraded Gaussian case, the source can still help the relay (which suffers from the interfering channel state), by using generalized dirty paper coding (GDPC),i.e., DPC combined with partial state cancellation. 
\section{Outline of Proof for Lemma \ref{AchievableRatePartiallyCooperativeRBCInformedSource}}\label{secV}
In the following, we denote the set of strongly jointly $\epsilon$-typical sequences (see \cite[Chapter 14.2]{CT91}) with distribution $p(x,y)$ as $T_{\epsilon}^n[\dv x,\dv y]$. We define $T_{\epsilon}^n[\dv x,\dv y|x^n]$ as
\begin{equation}
T_{\epsilon}^n[\dv x,\dv y|x^n]=\{y^n\::\:(x^n,y^n) \in T_{\epsilon}^n[\dv x,\dv y]\}.
\end{equation}
Note that it suffices to prove the result for the case without common message (i.e. $R_0=0$). This is because one can view part of the rate $R_2$ to be common rate $R_0$, since destination $1$ also decodes message $W_2$.

We assume that the source uses a combination of superposition coding \cite[Chapter 14.6]{CT91} and Gelfand and Pinsker's binning \cite{GP80}. We adopt the regular encoding/sliding window decoding strategy \cite{C82} for the decode-and-forward scheme. Decoding is based on a combination of joint typicality and sliding-window.

We consider a transmission over $B$ blocks, each with length $n$. A each of the first $B-1$ blocks, a pair of messages $(w_{1,i}, w_{2,i})\in \mscr W_1 \times \mscr W_2$ is sent, where $i$ denotes the index of the block, $i=1,\hdots,B-1$. For fixed $n$, the rate pair $(R_1\frac{B-1}{B},R_2\frac{B-1}{B})$ approaches $(R_1,R_2)$ as $B \longrightarrow +\infty$. We use random codes for the proof.

Fix a joint probability distribution of $U_1,U_2,X_1,X_2,S,Y_1,Y_2$ of the form $$p(s)p(u_1,u_2,x_1,x_2|s)p(y_1|x_1,x_2,s)p(y_2|y_1,x_2),$$ where $U_1$ and $U_2$ are two auxiliary random variables with bounded alphabet cardinality which stand for the information being carried by the source input $X_1$ and intended for destination $1$ and destination $2$, respectively.

Fix $\epsilon > 0$. Let
\begin{eqnarray}
J_1 &\triangleq& 2^{n(I(U_1;S|U_2X_2)+2\epsilon)},\nonumber\\
J_2 &\triangleq& 2^{n(I(U_2;S|X_2)+2\epsilon)},\nonumber\\
M_1 &\triangleq& 2^{n(R_1-4\epsilon)},\nonumber\\
M_2 &\triangleq& 2^{n(R_2-6\epsilon)}.\nonumber
\end{eqnarray}
\subsubsection*{Random codebook generation} We generate two statistically independent codebooks  (codebooks $1$ and $2$) by following the steps outlined below twice. These codebooks will be used for blocks with odd and even indices, respectively (see the {\it encoding} step).
\begin{enumerate}
\item[1.] Generate $M_2$ i.i.d. codewords $\dv x_2(w'')$, of length $n$ each, indexed by $w'' \in \{1,2,\hdots,M_2\}$, and each with distribution $\Pi_{i}p(x_{2i})$.
\item[2.] For each $\dv x_2(w'')$, generate a collection $\dv b(\dv x_2(w''))$ of $\dv u_2$-vectors
\begin{align}
&{\dv b}\Big(\dv x_2(w'')\Big)=\Big\{{\dv u_2}_{j_2,w'}({\dv x_2}(w'')),\:j_2 \in \{1,2,\cdots,J_2\},\nonumber\\
&\hspace{5cm} w' \in \{1,2,\cdots,M_2\}\Big\}\nonumber
\end{align}
independently of each other, each with distribution $\Pi_{i}p(u_{2i}|x_{2i}(w''))$.
\item[3.] For each $\dv x_2(w'')$, for each ${\dv u_2}_{j_2,w'}({\dv x_2}(w''))$, generate a collection $\dv a$ of $\dv u_1$-vectors
\begin{align}
{\dv a}\Big(\dv x_2(w''),&{\dv u_2}_{j_2,w'}({\dv x_2}(w''))\Big)=\Big\{{\dv u_1}_{j_1,w}({\dv u_2}_{j_2,w'}({\dv x_2}(w''))),\nonumber\\
&\hspace{0.5cm}j_1 \in \{1,2,\cdots,J_1\}, w \in \{1,2,\cdots,M_1\}\Big\}\nonumber
\end{align}
independently of each other, each with distribution $\Pi_{i}p(u_{1i}|u_{2i}(j_2,w'),x_{2i}(w''))$.  Reveal the collections $\dv a$ and $\dv b$ and the sequences $\{\dv x_2\}$ to the source and destinations $1$ and $2$.
\end{enumerate}
\subsubsection*{Encoding} We encode messages using codebooks $1$ and $2$, respectively, for blocks with odd and even indices. Using independent codebooks for blocks with odd and even indices makes the error events corresponding to these blocks independent and hence, the corresponding probabilities easier to evaluate.

At the beginning of block $i$, let $(w_{1,i},w_{2,i})$ be the new message pair to be sent from the source and $(w_{1,{i-1}},w_{2,{i-1}})$ be the pair sent in the previous block $i-1$.  Assume that at the beginning of block $i$, the relay has decoded $w_{2,{i-1}}$ correctly. The relay sends $\dv x_2(w_{2,i-1})$ . Given a state vector $\dv s=s^n$, let $j_2(\dv s,w_{2,i-1},w_{2,i})$ be the smallest integer $j_2$ such that
\begin{equation}
{\dv u_2}_{j_2,w_{2,i}}(\dv x_2(w_{2,i-1})) \in  T_{\epsilon}^n[\dv u_2,\dv x_2,\dv s|x_2^n].
\label{U2SGivenX2JointlyTypical}
\end{equation}
If such $j_2$ does not exist, set $j_2(\dv s,w_{2,i-1},w_{2,i})=J_2$. Sometimes, we will use $j_2^{\star}$ as shorthand for the chosen $j_2$. Let $j_1(\dv s,w_{2,{i-1}},w_{2,i},w_{1,i})$ be the smallest integer $j_1$ such that
\begin{align}
&\Big({\dv u_1}_{j_1,w_{1,i}}({\dv u_2}_{j_2^{\star},w_{2,i}}(\dv x_2(w_{2,i-1})) ), \dv s\Big)\nonumber\\
&\hspace{3cm} \in  T_{\epsilon}^n[\dv u_1,\dv u_2,\dv x_2, \dv s|u_2^n,x_2^n].
\label{U1SGivenU2X2JointlyTypical}
\end{align}
If such $j_1$ does not exist, set $j_1(\dv s,w_{2,{i-1}},w_{2,i},w_{1,i})=J_1$. Sometimes, we will use $j_1^{\star}$ as shorthand for the chosen $j_1$. Finally, generate a vector of input letters $\dv x_1 \in \mscr X_1^n$ according to the memoryless distribution defined by the $n-$product of
\begin{equation}
\Pi_ip(x_{1i}|u_{1i}(\dv u_2(\dv x_2)),u_{2i}(\dv x_2),s_i)
\end{equation}
\subsubsection*{Decoding} The decoding procedures at the end of block $i$ are as follows.
\begin{enumerate}
\item[1.] destination $1$, having known $w_{2,{i-1}}$, declares that $\hat{w}_{2,i}$ is sent if there is a unique $\hat{w}_{2,i}$ such that
\begin{align}
\Big({\dv u_2}_{j_2,\hat{w}_{2,i}}(\dv x_2(w_{2,i-1})),\dv y_1(i)\Big) \in  T_{\epsilon}^n[\dv u_2,\dv x_2, \dv y_1(i)|x_2^n].\nonumber
\end{align}
It can be shown that the decoding error in this step is small for sufficiently large $n$ if
\begin{equation}
R_2 < I(U_2;Y_1|X_2)-I(U_2;S|X_2).
\end{equation}
\item[2.] destination $1$, having known $w_{2,i-1}$ and $w_{2,i}$, declares  that the message $\hat{w}_{1,i}$ is sent if there is a unique $\hat{w}_{1,i}$ such that
\begin{align}
\Big({\dv u_1}_{j_1,\hat{w}_{1,i}}({\dv u_2}_{j_2,\hat{w}_{2,i}}(&\dv x_2(w_{2,i-1}))),\dv y_1(i)\Big)\nonumber\\
&\in  T_{\epsilon}^n[\dv u_1,\dv u_2,\dv x_2, \dv y_1(i)|x_2^n,u_2^n].\nonumber
\end{align}
It can be shown that the decoding error in this step is small for sufficiently large $n$ if
\begin{equation}
R_1 < I(U_1;Y_1|U_2X_2)-I(U_1;S|U_2X_2).
\end{equation}
\item[3.] Destination $2$ knows $w_{2,{i-2}}$ and decodes $w_{2,{i-1}}$ based on the information received in block $i-1$ and block $i$. It declares that the message $\hat{w}_{2,{i-1}}$ is sent if there is a unique $\hat{w}_{2,{i-1}}$ such that
\begin{subequations}
\begin{align}
&\bigg( \dv x_2(\hat{w}_{2,i-1}),\dv y_2(i)\bigg) \in T_{\epsilon}^n[\dv x_2,\dv y_2],\nonumber\\
&\bigg({\dv u_2}_{j_2,\hat{w}_{2,i-1}}(\dv x_2(w_{2,i-2})),\dv y_2(i-1) \bigg) \in T_{\epsilon}^n[\dv u_2,\dv x_2,\dv y_2|x_2^n].\nonumber
\end{align}
\end{subequations}
It can be shown that the decoding error in this step is small for sufficiently large $n$ if
\begin{equation}
R_2 < I(U_2X_2;Y_2)-I(U_2;S|X_2).
\end{equation}
\end{enumerate}
\bibliography{paperISIT2007}

\end{document}